\documentclass[12pt]{llncs}

\usepackage{amssymb,amsmath,stmaryrd}
\usepackage{pstricks, pst-node, pst-tree, pst-plot}
\usepackage{graphicx}
\usepackage{fullpage}
\usepackage{color}
\usepackage{times}

\def\N{{\mathbb{N}}}

\def\RP{{\mathbb{R^+}}}
\def\R{{\mathbb{R}}}
\def\opt{{\mathrm{O}\textsc{pt}}}

\def\tp{{\mathrm{tp}}}
\def\poa{\mathrm{PoA}}
\def\pos{\mathrm{PoS}}

\def\pay{u}
\def\P{\mathcal{P}}
\def\O{\mathcal{O}}
\def\ne{{\mathcal{NE}}}

\def\calG{{\mathcal G}}

\begin{document}

\title{Dynamics of Profit-Sharing Games}

\author{John Augustine, Ning Chen, Edith Elkind, Angelo Fanelli, Nick Gravin, Dmitry Shiryaev}
\institute{
Division of Mathematical Sciences \\
School of Physical and Mathematical Sciences\\ 
Nanyang Technological University, Singapore}

\maketitle

\begin{abstract}
An important task in the analysis of multiagent systems
is to understand how groups of selfish
players can form coalitions, i.e., work together in teams. 
In this paper, we study the dynamics of coalition formation 
under bounded rationality. 
We consider settings where each team's profit is given by a convex function, 
and propose three profit-sharing
schemes, each of which is based on the concept of marginal utility. 
The agents are assumed to be myopic, i.e., they keep changing teams
as long as they can increase their payoff by doing so.
We study the properties (such as closeness to Nash equilibrium
or total profit) of the states that result
after a polynomial number of such moves, and
prove bounds on the price of anarchy and the price of stability
of the corresponding games. 
\end{abstract}

\section{Introduction}\label{sec:intro}
Cooperation and collaborative task execution are fundamentally important both
for human societies and for multiagent systems. Indeed, it is often the case 
that certain tasks are too complicated or resource-consuming to be executed
by a single agent, and a collective effort is needed. Such settings are usually
modeled using the framework of {\em cooperative games}, which specify the 
amount of payoff that each subset of agents can achieve: when the game is played
the agents split into teams (coalitions), and the payoff of each team is divided 
among its members. 

The standard framework of cooperative game theory is static, i.e., it does not
explain how the players arrive at a particular set of teams and a payoff 
distribution. However, understanding the dynamics of coalition formation 
is an obviously important issue from the practical perspective, 
and there is an active stream of research that studies
bargaining and coalition formation in cooperative 
games (see, e.g.~\cite{chatterjee93,moldovanu95,okada96,yan2003}).
Most of this research assumes that the agents are fully rational, i.e., 
can predict the consequences of their actions and maximize their (expected)
utility based on these predictions.
However, full rationality is a strong assumption that is unlikely to
hold in many real-life scenarios: first, the agents may not have the computational 
resources to infer their optimal strategies, and second, they may not be 
sophisticated enough to do so, or lack information about other players. 
Such agents may simply respond to their current environment
without worrying about the subsequent reaction of other agents; such behavior
is said to be {\em myopic}. Now, coalition formation by computationally limited
agents has been studied by a number of researchers in multi-agent systems,  
starting with the work of \cite{shehory-kraus} and \cite{sandholm-lesser}. 
However, myopic behavior in coalition formation received relatively
little attention in the literature (for some exceptions, see
\cite{dieckmann,chalk,airiau}). In contrast, myopic dynamics of non-cooperative
games is the subject of a growing body of research
(see, e.g.~\cite{FabrikantPT04,AwerbuchAEMS08,fanelli08}). 

In this paper, we merge these streams of research and apply techniques developed in the context
of analyzing the dynamics of non-cooperative games to coalition formation settings.
In doing so, we depart from the standard model of games with transferable utility, which
allows the players in a team to share the payoff arbitrarily: indeed, 
such flexibility will necessitate a complicated negotiation process
whenever a player wants to switch teams. Instead, we consider three 
payoff models that are based on the concept of marginal utility, i.e., 
the contribution that the player makes to his current team. Each of the payoff
schemes, when combined with a cooperative game, induces a non-cooperative
game, whose dynamics can then be studied using the rich set of tools
developed for such games in recent years.

We will now describe our payment schemes in more detail. We assume that
we are given a convex cooperative game, i.e., the values of the teams are 
given by a submodular function; the submodularity property 
means that a player is more useful when he joins a smaller team, and 
plays an important 
role in our analysis. In our first scheme, the payment to each agent is given by
his marginal utility for his current team; by submodularity, the total
payment to the team members never exceeds the team's value.
This payment scheme rewards
each agent according to the value he creates; we will therefore
call these games {\em Fair Value games}.
Our second scheme takes into account the history of the interaction: 
we keep track of the order
in which the players have joined their teams, and pay each agent his marginal contribution
to the coalition formed by the players who joined his current team before him. This ensures
that the entire payoff of each team is fully distributed among its members. Moreover, 
due to the submodularity property a player's payoff never goes down as long as he stays
with the same team. This payoff scheme is somewhat reminiscent of the reward
schemes employed in industries with strong labor unions; we will therefore refer
to these games as {\em Labor Union games}. Our third scheme can be viewed
as a hybrid of the first two: it distributes the team's payoff
according to the players' Shapley values,
i.e., it pays each player his expected marginal contribution
to a coalition formed by its predecessors when players are reordered randomly;
the resulting games are called {\em Shapley games}.

\smallskip

\noindent{\bf Our contributions\quad} We study the equilibria and dynamics of the
three games described above. We are interested in the properties of the states
that can be reached by natural dynamics in a polynomial number of steps: in 
particular, whether such states are (close to) Nash equilibria, and whether
they result in high total productivity, i.e., the sum of the teams' values
(note that in Fair Value games the latter quantity may differ from the social 
welfare, i.e., the sum of players' payoffs).

We first show that all our games are
potential games, and hence admit a Nash equilibium in pure strategies. 
We then argue that for each of our games the price of anarchy is bounded by 2.
For the first two classes of games, we can also bound their $\alpha$-price
of anarchy, i.e., the ratio between the total profit of the optimal coalition structure
and that of the worst $\alpha$-Nash equilibrium, by $2+\alpha$.  
We also provide bounds on the price of stability for all three games.
Further, for the first two 
classes of games, we show that the basic Nash dynamic converges in polynomial time
to an approximately optimal state, where the approximation ratio is arbitrarily close to the 
price of anarchy; 
these results extend to basic $\alpha$-Nash dynamic and $\alpha$-price of anarchy.
To obtain these results, we observe that both the Fair Value games and the
Labor Union games can be viewed as variants of $\beta$-nice games
introduced in~\cite{AwerbuchAEMS08}, and prove general convergence results for such games, 
which may be of independent interest.
We then show that Labor Union games have additional desirable properties:
in such games $\alpha$-Nash dynamics
quickly converges to $\alpha$-Nash equilibrium; also, if we start with the state where
each player is unaffiliated, the Nash dynamics converges to a Nash equilibrium 
after each player gets a chance to move.

The rest of the paper is organized as follows. After a brief overview of the related
work, we provide the required preliminaries in Section~\ref{sec:prelim}. 
Section~\ref{sec:beta} deals with $\beta$-nice games and lays the groundwork  
that will be necessary for the technical results in the next section. Then, 
in Section~\ref{sec:main}, we describe our three classes of games and
present our results for these games. Section~\ref{sec:cut} explains the 
relationship between our games and the well-studied cut games.
Section~\ref{sec:concl} presents our conclusions and directions for future work.

\smallskip

\noindent{\bf Related Work\quad} 
The games studied in this paper belong to the class of {\em potential games}, 
introduced by Monderer and Shapley~\cite{MondererShapley}. 
In potential games,
any sequence of improvements by players converges to  a pure Nash equilibrium. 
However, the number of steps can be exponential in the description of the game. 
The complexity of computing (approximate) 
Nash equilibrium in various subclasses of potential games
such as congestion games~\cite{R73}, cut games~\cite{Yann91} or party affiliation 
games~\cite{FabrikantPT04} has received a lot of attention in recent 
years~\cite{Johnson88,FabrikantPT04,ChristodoulouMS06,SkopalikV08,Tob10,Bhalgat10}.
A related issue is how long it takes for some form of best
response dynamics to reach an
equilibrium~\cite{mir04,goemans05,chien07,ack08,SkopalikV08,AwerbuchAEMS08}.
Even if a Nash equilibrium cannot be reached quickly, 
a state reached after a polynomial number of steps may still have high social welfare;
this question is studied, for example, in~\cite{ChristodoulouMS06,fanelli08,fanelli09}.

A recent paper by Gairing and Savani~\cite{gairing10} studies the dynamics
of a class of cooperative games known as additively separable hedonic games;
their focus is on the complexity of computing stable outcomes. While the class
of all convex cooperative games considered in this paper is considerably
broader than that of additively separable games, paper~\cite{gairing10}
also studies notions of stability not considered here.

\section{Preliminaries}\label{sec:prelim}
\noindent{\bf Non-cooperative games.\quad} A \emph{non-cooperative
game} is defined by a tuple ${\cal G} =\left(N, (\Sigma_i)_{i \in
N}, (\pay_i)_{i \in N}\right)$, where $N=\{1,2,\ldots,n\}$ is the
set of {\em players}, $\Sigma_i$ is the set of (pure) {\em
strategies} of player $i$, and $\pay_i : \times_{i \in N}\Sigma_i
\to \RP \cup \{0\}$ is the {\em payoff function} of player $i$.

Let $\Sigma = \times_{i \in N} \Sigma_i$ be the \emph{strategy
profile set} or \emph{state set} of the game, and let $S = (s_1,
s_2,\ldots, s_n) \in \Sigma$ be a generic state in which each player
$i$ chooses strategy $s_i \in \Sigma_i$. Given a strategy profile
$S = (s_1, s_2,\ldots, s_n)$ and a strategy $s'_i \in \Sigma_i$,
let $(S_{-i}, s'_i)$ be the strategy profile obtained from $S$ 
by changing the strategy of player $i$ from $s_i$ to $s'_i$, i.e., 
$(S_{-i}, s'_i) = (s_1, s_2,\ldots, s_{i-1}, s_{i}', s_{i+1}, \ldots,s_n)$.

\smallskip

\noindent{\bf Nash equilibria and dynamics. \quad} 
Given a strategy profile $S =
(s_1, s_2, \ldots, s_n)$, a strategy $s'_i\in\Sigma_i$ is an
\emph{improvement move} for player $i$ if $\pay_i(S_{-i}, s'_i) >
\pay_i(S)$; further, $s'_i$ is called an {\em $\alpha$-improvement
move} for $i$ if $\pay_i(S_{-i}, s'_i) > (1+\alpha)\pay_i(S)$, where
$\alpha>0$. A strategy $s^b_i \in \Sigma_i$ is a \emph{best
response} for player $i$ in state $S$ if it yields the maximum
possible payoff given the strategy choices of the other players,
i.e., $\pay_i(S_{-i}, s^b_i) \geq \pay_i(S_{-i}, s'_i)$ for any
$s'_i \in \Sigma_i$. 
An {\em $\alpha$-best
response move} is both an $\alpha$-improvement and a best response
move.

A (pure) \emph{Nash equilibrium} is a strategy profile in which
every player plays her best response. Formally, $S = (s_1, s_2,
\ldots, s_n)$ is a Nash equilibrium if for all $i \in N$ and for any
strategy $s'_i \in \Sigma_i$ we have $\pay_i(S) \geq \pay_i(S_{-i}, s'_i)$.
We denote the set of all (pure) Nash equilibria of a game ${\cal G}$
by $\ne({\cal G})$. 
A profile $S=(s_1, \dots, s_n)$
is called an {\em $\alpha$-Nash equilibrium} if 
no player can improve his payoff by more than a factor
of $(1+\alpha)$ by deviating, i.e., 
$(1+\alpha)u_i(S)\ge u_i(S_{-i}, s'_i)$
for any $i\in N$ and any $u'_i\in \Sigma_i$.
The set of all  $\alpha$-Nash equilibria of $\calG$
is denoted by $\ne^\alpha({\cal G})$.
In a {\em strong Nash equilibrium}, no group of players can improve
their payoffs by deviating, i.e., $S=(s_1, \dots, s_n)$ is 
a strong Nash equilibrium if for all $I\subseteq N$ and
any strategy vector $S'=(s'_1, \dots, s'_n)$ such that 
$s'_i=s_i$ for $i\in N\setminus I$, if
$\pay_i(S')> \pay_i(S)$ for some $i\in I$, then
$\pay_j(S')< \pay_j(S)$ for some $j\in I$.
%
%

Let $\Delta_i(S)$ be the improvement in the player's payoff if he performs
his best response, i.e., $\Delta_i(S) = \pay_i(S_{-i}, s^b_i) - \pay_i(S)$,
where $s^b_i$ is the best response of player $i$ in state $S$.
For any $Z \subseteq N$ let  $\Delta_{Z}(S) = \sum_{i \in Z} \Delta_i (S)$,
and let $\Delta(S) = \Delta_N (S)$.
A \emph{Nash dynamic} (respectively, \emph{$\alpha$-Nash dynamic})
is any sequence of best response (respectively, $\alpha$-best response) moves.
A \emph{basic Nash dynamic} (respectively, \emph{basic $\alpha$-Nash dynamic})
is any Nash dynamic (respectively, $\alpha$-Nash dynamic)
such that at each state $S$
the player $i$ that makes a move has the maximum absolute improvement, i.e.,
$i\in \arg\max_{j \in N} \Delta_j(S)$.

\smallskip

\noindent{\bf Price of anarchy.\quad}Given a game $\calG$ with a set
of states $\Sigma$, and a function $f:\Sigma\to\RP\cup\{0\}$, we
write $\opt_f(\calG)=\max_{S\in\Sigma}f(S)$.
The \emph{price of anarchy} $\poa_f(\calG)$ and the \emph{price of stability}
$\pos_f(\calG)$ of a game $\cal G$ with respect to a function $f$ are, 
respectively, the worst-case ratio and the best-case ratio
between the value of $f$ in a Nash equilibrium and $\opt_f(\calG)$, i.e.,
$\poa_f({\cal G})=\max_{S\in\ne(\calG)}\frac{\opt_f(\calG)}{f(S)}$,
$\pos_f({\cal G})=\min_{S\in\ne(\calG)}\frac{\opt_f(\calG)}{f(S)}$.
The {\em strong price of anarchy} and the {\em strong price of stability}
are defined similarly; the only difference is that the maximum 
(respectively, minimum) is taken over all strong Nash equilibria.
Further, the {\em $\alpha$-price of anarchy}  $\poa^\alpha_f(\calG)$ 
of a game $\calG$ with respect to $f$ is defined as 
$\poa^\alpha_f(\calG)=\max_{S\in\ne^\alpha(\calG)}\frac{\opt_f(\calG)}{f(S)}$;
the {\em $\alpha$-price of stability} $\pos^\alpha_f(\calG)$ 
can be defined similarly.
Originally, these notions were defined with respect to the social welfare function, 
i.e., $f=\sum_{i \in N} u_i(S)$. 
However, we give a more general definition
since in the setting of this paper it is natural
to use a different function $f$. We omit the index $f$ when the function $f$
is clear from the context.

\smallskip

\noindent{\bf Potential games.\quad} A non-cooperative game $\calG$
is called a \emph{potential game} if there is a function $\Phi :
\Sigma\to\N$ such that for any state $S$ and any improvement move
$s'_i$ of a player $i$ in $S$ we have  $\Phi(S_{-i}, s'_i) - \Phi(S)
> 0$; the function $\Phi$ is called the \emph{potential function} of
$\calG$. The game $\calG$ is called an \emph{exact potential game}
if we have $\Phi(S_{-i}, s'_i) - \Phi(S) = \pay_i(S_{-i}, s'_i) -
\pay_i(S)$. It is known that any potential game has a pure Nash
equilibrium~\cite{MondererShapley,R73}.

\smallskip

\noindent{\bf Cooperative games.\quad} A {\em cooperative} game
$G=(N, v)$ is given by a set of {\em players} $N$ and a {\em
characteristic function} $v:2^N\to\RP\cup\{0\}$ that for each set
$I\subseteq N$ specifies the profit that the players in $I$ can
earn by working together. We assume that $v(\emptyset)=0$. A {\em
coalition structure} over $N$ is a partition of players in $N$,
i.e., a collection of sets $I_1, \dots, I_k$ such that (i)
$I_i\subseteq N$ for $i=1, \dots, k$; (ii) $I_i\cap I_j=\emptyset$
for all $i< j\le k$; and (iii) $\cup_{j=1}^k I_j=N$. A game $G=(N,
v)$ is called {\em monotone} if $v$ is {\em non-decreasing}, i.e.,
$v(I)\le v(J)$ for any $I\subset J\subseteq N$. Further, $G$ is
called {\em convex} if $v$ is {\em submodular}, i.e., for any
$I\subset J\subseteq N$ and any $i\in N\setminus J$ we have
$v(I\cup\{i\})-v(I)\ge  v(J\cup\{i\})-v(J)$. Informally, in a convex
game a player is more useful when he joins a smaller coalition. 
We will make use of the
following property of submodular functions.

\begin{lemma}\label{lem_sub}
Let  $f: 2^V \to \R$ be a submodular function. Then for any pair of sets $X, Y\subseteq V$ 
such that $X \cap Y = \emptyset$ and $X=\{x_1, x_2,\ldots, x_k\}$, it holds that
$\sum_{j=1, \ldots, k} \left(f(Y \cup \{x_j\}) - f(Y) \right)\geq 
f(Y \cup X) - f(Y)$. 
\end{lemma}
\begin{proof}
Since $f$ is a submodular function, for every $x_j \in X$ we have
$$
f(Y\cup \{x_j\}) - f(Y)  \geq f(Y \cup \{x_1, x_2, \ldots, x_{j-1}, x_j\}) - f(Y \cup \{x_1,
x_2, \ldots, x_{j-1}\}). 
$$
The lemma now follows by summing these inequalities for all $j=1,\ldots, k$. 
\qed
\end{proof}

\section{Perfect $\beta$-nice Games}\label{sec:beta}
In this section, we define the class of perfect $\beta$-nice games
(our definition is inspired by~\cite{AwerbuchAEMS08}, but differs
from the one given there), 
and prove a number of results for such games. Subsequently, we will show
that many of the profit-sharing games considered in the paper belong 
to this class. Most proofs in this section are relegated to 
Appendix~\ref{app:beta}.
\begin{definition}\label{def:beta-nice}
A potential game $\calG$ with a potential function $\Phi$ is called {\em perfect}
with respect to a function $f:\Sigma\to\RP\cup\{0\}$ if for any 
state $S$ it holds that $f(S) \geq \sum_{i\in N}u_i(S)$, and, moreover,  
%
%
for any improvement move $s'_i$ of player $i$ we have 
$$
f(S_{-i}, s'_i) - f(S)  \geq 
\Phi(S_{-i}, s'_i) - \Phi(S) \geq 
\pay_i(S_{-i}, s'_i) - \pay_i(S).
$$ 
Also, a game $\calG$ is called {\em $\beta$-nice} with respect to $f$ 
if for every state $S$ we have
$\beta \cdot f(S)  +  \Delta(S)  \geq  \opt_f(\calG)$.
\end{definition}
We can bound the price of anarchy of a $\beta$-nice game by $\beta$.

\begin{lemma}\label{lem_poa}
For any $f:\Sigma\to\RP\cup\{0\}$ and
any game $\calG$ that is $\beta$-nice w.r.t. $f$ we have $\poa_f(\calG) \le \beta$.
\end{lemma}
\begin{proof}
The lemma follows by observing that for any Nash equilibrium $S$ we have 
$\Delta(S) \leq 0$.
\qed
\end{proof}
Lemma~\ref{lem_poa} can be extended to $\alpha$-price of anarchy
for any $\alpha\ge 0$.

\begin{lemma}\label{lem_poa_alpha}
For any $f:\Sigma\to\RP\cup\{0\}$, any $\alpha\ge 0$, and
any game $\calG$ that is $\beta$-nice w.r.t. $f$ we have $\poa^\alpha_f(\calG) 
\le \alpha+\beta$.
\end{lemma}
\begin{proof}
For any $\alpha$-Nash equilibrium $S$ we have
$\Delta(S) \leq \alpha\sum_{i\in N}u_i(S)\leq \alpha f(S)$.
\qed
\end{proof}
We now state a technical lemma that we use shortly in proving Theorem~\ref{TH_1}.
\begin{lemma}\label{LM_1}
Consider any non-cooperative game $\calG$ and any function
$f:\Sigma\to\RP\cup\{0\}$. For positive values of $\epsilon$, $a$, and $b$, 
any dynamic for which the increase in the value of $f$ at a step
leading from $S$ to $\bar{S}$ is at least $b - \frac{1}{a}f(S)$
converges to a state $S^F$ with  $f(S^F)  \geq ab(1 - \epsilon)$ in
at most $\big\lceil a\ln\frac{1}{\epsilon}\big\rceil$ steps, from
any initial state.
\end{lemma}
The next theorem states that after a polynomial number of steps, for every perfect 
$\beta$-nice potential game, the basic Nash dynamic reaches a state whose 
relative quality 
(with respect to $f$) is close to the price of anarchy.

\begin{theorem}\label{TH_1}
Consider any function $f:\Sigma\to\RP\cup\{0\}$ and any game $\calG$
that is perfect $\beta$-nice with respect to $f$. For any
$\epsilon>0$ the basic Nash dynamic converges to a state $S^F$ with
$f(S^F) \geq \frac{\opt_f(\calG)}{\beta}(1 - \epsilon)$ in at most 
$\big\lceil\frac{n}{\beta}\ln\frac{1}{\epsilon}\big\rceil$ steps,
starting from any initial state.
\end{theorem}
\begin{proof}
Consider a generic state $S$ of the dynamic. 
Since $\calG$ is $\beta$-nice,
we have
$\Delta(S)  \geq  \opt_f(\calG) - \beta \cdot f(S)$.
Let $i$ be the player moving in state $S$, and let $\bar{S}$ be the state resulting
from the move of player $i$. Since $i$ is the player with the maximum absolute improvement,
we get
$$
 f(\bar{S}) - f(S)  \geq   \Phi(\bar{S}) - \Phi(S) 
 \geq  \Delta_i(S) 
  \geq  \frac{\Delta(S)}{n}  
  \geq   \frac{\opt_f(\calG) - \beta\cdot f(S)}{n}.
$$
The theorem now follows by applying Lemma~\ref{LM_1} with $b=\frac{\opt_f(\calG)}{n}$ and 
$a=\frac{n}{\beta}$.
\qed
\end{proof}
A convergence result similar to Theorem~\ref{TH_1} 
can be obtained for basic $\alpha$-Nash dynamic.

\begin{theorem}\label{TH_2}
Consider any function $f:\Sigma\to\RP\cup\{0\}$ and any game $\calG$
that is perfect $\beta$-nice with respect to $f$. For any
$\epsilon>0$ and any $\alpha\ge 0$ the basic $\alpha$-Nash dynamic
converges to a state $S^F$ with $f(S^F) \geq
\frac{\opt_f(\calG)}{\beta+\alpha}(1 - \epsilon)$ in at most
$\big\lceil\frac{n}{\beta+\alpha}\ln\frac{1}{\epsilon}\big\rceil$
steps, starting from any initial state.
\end{theorem}

\section{Profit-sharing games}\label{sec:main} 
In this section, we study three non-cooperative games that
can be constructed from an arbitrary monotone convex cooperative game. 

Each of our games can be described by a triple $\calG=(N, v, M)$, 
where $(N, v)$ is a monotone convex 
cooperative game with $N=\{1, \dots, n\}$,
and $M=\{1, \dots, m\}$ is a set of $m$ parties; we require $m\le n$. 
All three games considered
in this section model the setting where the players in $N$ 
form a coalition structure over $N$ that consists of $m$ coalitions.
Thus, each player needs to choose exactly one party from $M$, i.e., 
for each $i\in N$ we have $\Sigma_i=M$. In some cases (see Section~\ref{sec:LU}), 
we also allow players to be unaffiliated. To model this, 
we expand the set of strategies by setting $\Sigma_i=M\cup\{0\}$. 
Intuitively, the parties correspond to different companies, and the players
correspond to the potential employees of these companies; we desire 
to assign employees to companies so as to maximize the total productivity.

In two of our games 
(see Section~\ref{sec:FM} and Section~\ref{sec:SH}), a state
of the game is completely described by the assignment of the players
to the parties, i.e., we can write $S=(s_1, \dots, s_n)$, 
where $s_i\in M$ for all $i\in N$. Alternatively, we can specify
a state of the game by providing a partition of the set $N$
into $m$ components $Q_1, \dots, Q_m$, where $Q_j$ is the set
of all players that chose party $j$, i.e., we can write $S = (Q_1, \dots, Q_m)$;
we will use both forms of notation throughout the paper.
In the game described in Section~\ref{sec:LU}, the state of the game depends not only
on which parties the players chose, but also on the order in which 
they joined the party; we postpone the formal description of this
model till Section~\ref{sec:LU}. In all three models, 
each player's payoff is based on the concept of marginal utility; 
however, in different models this idea is instantiated in different ways.

An important parameter of a state $S=(Q_1, \dots, Q_m)$ in each of these games is 
its {\em total profit} $\tp(S)=\sum_{j\in M}v(Q_j)$. While for the games
defined in Section~\ref{sec:LU} and Section~\ref{sec:SH}, the total profit
coincides with the social welfare, for the game described
in Section~\ref{sec:FM} this is not necessarily the case.
As we are interested in finding the most efficient partition of players 
into teams, we consider the total profit of a state a more relevant 
quantity than its social welfare. Therefore, 
in what follows, we will consider the price of anarchy and the price of
stability with respect to the total profit, i.e., 
we have $\opt(\calG)=\opt_\tp(\calG)$, $\poa(\calG)=\poa_\tp(\calG)$, 
$\pos(\calG)=\pos_\tp(\calG)$. 

All of our results generalize to the setting where each party $j\in M$
is associated with a different non-decreasing submodular profit function 
$v_j:2^N\to\RP\cup\{0\}$,
i.e., different companies possess different technologies, and therefore may
have different levels of productivity.
Formally, any such game is given by a tuple $\calG=(N, v_1, \dots, v_m, M)$, 
where $M=\{1, \dots, m\}$, and for each $j\in M$ the function 
$v_j$ is a non-decreasing submodular function $v_j:2^N\to\RP\cup\{0\}$ 
that satisfies $v(\emptyset)=0$.
In this case, 
the total profit function in a state $S=(Q_1, \dots, Q_m)$ is given by
$\tp(S)=\sum_{j\in M}v_j(Q_j)$.
In what follows, we present our results for this more general setting.

\subsection{Fair Value games}\label{sec:FM}
In our first model, the utility $u_i(S)$ of a player $i$ in a state
$S=(Q_1, \dots, Q_m)$ is given by $i$'s marginal contribution to
the coalition he belongs to, i.e., if $i\in Q_j$, we set
$u_i(S)=v_j(S)-v_j(S\setminus\{i\})$. As this payment scheme
rewards each player according to the value he creates, 
we will refer to this type of games as {\em Fair Value games}.
Observe that since the functions $v_j$
are assumed to be submodular, 
we have $\sum_{i\in Q_j}u_i(S)\le v_j(Q_j)$ for all $j\in M$, i.e., the total
payment to the employees of a company never exceeds the profit
of the company. Moreover, it may be the case that the profit 
of a company is strictly greater than the amount it pays to its employees;
we can think of the difference between the two quantities as the owner's/shareholders'
value. Consequently, in these games
the total profit of all parties may differ from the social
welfare, as defined in Section~\ref{sec:prelim}. 

We will now argue that 
Fair Value games have a number of desirable properties.
In particular, any such game is a potential game, 
and therefore has a pure Nash equilibrium. The proof of the following theorem
can be found in Appendix~\ref{app:FM}.

\begin{theorem}\label{th_2}
Every Fair Value game $\calG$ is a perfect $2$-nice exact potential game
w.r.t. the total profit function.
\end{theorem}
Combining Theorem~\ref{th_2}, Lemmas~\ref{lem_poa} and~\ref{lem_poa_alpha} 
and Theorems~\ref{TH_1} and~\ref{TH_2}, 
we obtain the following corollaries.

\begin{corollary}
For every Fair Value game $\calG$ and every $\alpha\ge 0$
we have $\poa^\alpha(\calG) \le 2+\alpha$. In particular, 
$\poa(\calG) \le 2$.
\end{corollary}

\begin{corollary}
For every Fair Value game $\calG$ and any $\epsilon>0$, the basic
Nash dynamic (respectively, the basic $\alpha$-Nash dynamic)
converges to a state $S^F$ with total profit $\tp(S^F) \geq
\frac{\opt(\calG)}{2}(1 - \epsilon)$ (respectively, $\tp(S^F) \geq
\frac{\opt(\calG)}{2+\alpha}(1 - \epsilon)$) in at most $\big\lceil
\frac{n}{2}\ln\frac{1}{\epsilon}\big\rceil$ steps
(respectively, $\big\lceil
\frac{n}{2+\alpha}\ln\frac{1}{\epsilon}\big\rceil$
steps), from any initial state.
\end{corollary}
Since every Fair Value game is an exact potential game
with the potential function given by the total profit, any 
profit-maximizing state is necessarily a Nash equilibrium.
This implies the following proposition.
\begin{proposition}
For any Fair Value game $\calG$ we have $\pos(\calG)=1$.
\end{proposition}
\subsection{Labor Union Games}\label{sec:LU}
In Fair Value games, the player's payoff only depends on his
current marginal value to the enterprise, i.e., one's salary may go
down as the company expands. However, in many real-life settings,
this is not the case. For instance, in many industries, especially
ones that are highly unionized, an employee that has spent many
years working for the company typically receives a higher salary
than a new hire with the same set of skills. Our second class of
games, which we will refer to as {\em Labor Union games}, aims to
model this type of settings. Specifically, in this class of games,
we modify the notion of state so as to take into account the order
in which the players have joined their respective parties; the
payment to each player is then determined  by his marginal utility
{\em for the coalition formed by his predecessors}. The
submodularity property guarantees than a player's payoff never goes
down as long as he stays with the same party.


Formally, in a Labor Union game $\calG$ that corresponds to a tuple 
$(N, v_1, \dots, v_m, M)$, 
we allow the players to be unaffiliated, i.e., 
for each $i\in N$ we set $\Sigma_i=M\cup\{0\}$. If player $i$ plays strategy $0$, 
we set his payoff to be $0$ irrespective of the other players' strategies.
A {\em state} of $\calG$ is given by a tuple
$\P = (P_1, \dots, P_m)$, where $P_j$ is the sequence of players in party $j$,  
ordered according to their arrival time. 
As before, the profit of party $j$ is given by the function $v_j$; 
note that the value of $v_j$ does not  
depend on the order in which the players join $j$. The payoff of each player, 
however, is dependent on their position in the affiliation order. 
Specifically, for a player $i \in P_j$, let 
$P_j(i)$ be the set of players that appear in $P_j$ before $i$.
Player $i$'s payoff is then defined  
as $\pay_i(\P) = v_j(P_j(i) \cup \{i\}) - v_j(P_j(i))$. 

We remark that, technically speaking, Labor Union games are not
non-cooperative games. Rather, each state of a Labor Union game
induces a non-cooperative game as described above; after any player
makes a move, the induced non-cooperative game changes. Abusing
terminology, we will say that a state $\P$ of a Labor Union game
$\calG$ is a Nash equilibrium if for each player $i\in N$ staying
with his current party is a best response in the induced game; all
other notions that were defined for non-cooperative games in
Section~\ref{sec:prelim}, as well as the results in
Section~\ref{sec:beta}, can be extended to Labor Union games in a
similar manner.

We now state two fundamental properties of our model. 
\begin{itemize}
\item
Guaranteed payoff: Consider two players $i$ and $i'$ in $P_j$.
Suppose $i'$ moves to another party. The payoff of player $i$ will
not decrease. Indeed, if $i'$ succeeds $i$ in the sequence $P_j$,
then by definition, $i$'s payoff is unchanged. If $i'$ precedes $i$
in $P_j$, then, since $v_j$ is non-decreasing and submodular, $i$'s
payoff will not decrease; it may, however, increase.

\item
Full payoff distribution: The sum of the payoffs of players within a
party $j$ is a telescopic sum that evaluates to $v_j(P_j)$.
Therefore, the total profit $\tp(\P)=\sum_{j\in M}v_j(P_j)$ in a
state $\P$ equals to the social welfare in this state. In other
words, in Labor Union games, the profit of each enterprise is
distributed among its employees, without creating any value for the
owners/shareholders.
%
%
\end{itemize}
The guaranteed payoff property distinguishes the Labor Union games
from the Fair Value games, where 
a player who maintains his affiliation
to a party might not be rewarded, but may rather see a reduction in his payoff  
as other players move to join his party. This, of course,
may incentivize him to shift his affiliation as well,
leading to a vicious cycle of moves. In contrast, 
in Labor Union games, a player is guaranteed that
his payoff will not decrease if he maintains his affiliation to a party.
This suggests that in Labor Union games stability may be
easier to achieve. In what follows, we will see that this 
is indeed the case.

We will first show that Labor Union games are perfect $2$-nice
with respect to the total profit (or, equivalently, social welfare);
this will allow us to apply the machinery developed in 
Section~\ref{sec:beta}.
Abusing notation, let
$\Delta_i(\P)$ denote the improvement in the payoff of player $i$ if he performs a
best response move from $\P$, and let $\Delta(\P) = \sum_{i\in N} \Delta_i(\P)$.
\begin{proposition}\label{prop:LU-nice}
Any Labor Union game $\calG$ is a perfect 2-nice game
with respect to the total profit function. 
\end{proposition}
\begin{proof}
It is easy to see that $\calG$ is
a potential game with the potential function $\Phi(\P)=\tp(\P)$.
Furthermore, for any player $i$ the increase in his payoff when he performs
an improvement move does not exceed the change in the total profit.  
It remains to show that $2 \tp(\P) + \Delta(\P) \ge \opt(\calG)$
for any $\P=(P_1, \dots, P_m)$.
We have
\begin{align*}   
v_j(O_j) \le v_j(P_j \cup O_j) = v_j(P_j) + v_j(P_j \cup O_j) -  v_j(P_j) 
         \le  v_j(P_j) + \sum_{i \in O_j\setminus P_j} (u_i(P_i) + \Delta_i(\P)).
\end{align*}
Summing over all parties, we obtain
$$
\opt(\calG) = \sum_{j \in M}  v_j(O_j) \le 
\sum_{j\in M}v_j(P_j)+\sum_{j\in M}\sum_{i \in O_j\setminus P_j} u_i(P_i)+
\sum_{j\in M}\sum_{i \in O_j\setminus P_j} \Delta_i(P)
\le 2\tp(\P) + \Delta(\P).
$$
\qed
\end{proof}
As in the case of Fair Value games, Proposition~\ref{prop:LU-nice}
allows us to bound the price of anarchy of any Labor Union game, 
as well as the time it takes to converge to a state with a ``good''
total profit.

\begin{corollary}
For every Labor Union game $\calG$ and every $\alpha\ge 0$
we have $\poa^\alpha(\calG) \le 2+\alpha$. In particular,
$\poa(\calG) \le 2$.
\end{corollary}

\begin{corollary}
For every Labor Union game $\calG$ and any $\epsilon>0$, the basic
Nash dynamic (respectively, the basic $\alpha$-Nash dynamic)
converges to a state $S^F$ with total profit $\tp(S^F) \geq
\frac{\opt(\calG)}{2}(1 - \epsilon)$ (respectively, $\tp(S^F) \geq
\frac{\opt(\calG)}{2+\alpha}(1 - \epsilon)$) in at most $\big\lceil
\frac{n}{2}\ln\frac{1}{\epsilon}\big\rceil$ steps
(respectively, $\big\lceil
\frac{n}{2+\alpha}\ln\frac{1}{\epsilon}\big\rceil$
steps), from any initial state.
\end{corollary}

Let $\O(\calG) = (O_1, \dots, O_m)$ 
be a state that maximizes the total profit in a game $\calG$, 
and let $\opt(\calG)=\tp(\O(\calG))$. 
As in the case of Fair Value games, it is not hard to see
that $\O(\calG)$ is a Nash equilibrium, i.e., $\pos(\calG)=1$.
In fact, for Labor Union games, we can prove a stronger statement.

\begin{proposition}\label{obs:PoS}
In any Labor Union game $\calG$,  
$\O(\calG)$ is a strong Nash equilibrium. 
I.e., the strong price of stability is 1.
\end{proposition}
\begin{proof}
Consider a deviating coalition $I\subseteq N$. By the guaranteed
payoff property, the deviation does not lower the payoff of all
players in $N\setminus I$ and increases the payoff of some of the
deviators, without harming the rest  of the deviators. Thus, the
deviation must lead to a state whose total payoff exceeds that of
$\O(\calG)$, a contradiction. \qed
\end{proof}
Furthermore, for Labor Union games we can show that for certain
dynamics and certain initial states one can guarantee convergence
to $\alpha$-Nash equilibrium or even Nash equilibrium.

\begin{proposition}\label{obs:fast}
Consider any Labor Union game $\calG=(N, v_1, \dots, v_m, M)$ 
such that $v_j(I)\ge 1$ for any $j\in M$ and any $I\in 2^N\setminus\{\emptyset\}$.
For any such $\calG$, the $\alpha$-Nash dynamic 
starting from any state in which all players are affiliated with some party
converges to an $\alpha$-Nash equilibrium in $O(\frac{n}{\alpha} \log W)$ steps. 
where $W$ is the maximum payoff that any player can achieve.
\end{proposition}
\begin{proof}
After each move in the $\alpha$-Nash dynamic, a player improves her
payoff by a factor of $1 + \alpha$, and the guaranteed payoff
property ensures that payoffs of other players are unaffected. So,
if a player starts with a payoff of at least $1$, she will reach a
payoff of $W$ after $O(\frac{\log W}{\alpha})$ steps. Therefore, in
$O(\frac{n}{\alpha} \log W)$ steps, we are guaranteed to reach an
$\alpha$-Nash equilibrium. \qed
\end{proof}

\begin{proposition}
Suppose a Labor Union game $\calG$ with $n$ players
starts at a state in which every 
player is unaffiliated. Then, in exactly $n$ steps of the Nash dynamic, 
the system will reach a Nash equilibrium. 
\end{proposition}
\begin{proof}
The proof is by induction on the number of steps. The very first
player who gets to move will pick the party that maximizes her
payoff. Subsequently, she will never have an incentive to move,
because no move will give her any improvement in her payoff. For the
inductive step, suppose that $k-1$ steps have elapsed, and exactly
$k-1$ players have moved once each and have reached their final
destination with no incentive to move again. The player who moves at
step $k$ chooses his best response party. Since the profit functions
are increasing and submodular, he cannot improve his payoff by
moving to another party at a later step.  Therefore, in $n$ steps,
the system reaches a Nash equilibrium. \qed
\end{proof}

We conclude with an important open question. We have shown that for
$\alpha>0$, the $\alpha$-Nash dynamic leads to an $\alpha$-Nash
equilibrium in $O(\frac{n}{\alpha} \log W)$ steps. However, we do
not know whether there exists a dynamic that converges to a Nash
equilibrium in a number of steps that is a polynomial in $n$ and
$\log W$.

\subsection{Shapley games}\label{sec:SH}
In our third class of games, which we call {\em Shapley games}, the
players' payoffs are determined in a way that is inspired by the
definition of the Shapley value~\cite{shap53}. Like in Fair Value
games, a state of a Shapley game is fully described by the partition
of the players into parties. Given a state $S=(Q_1, \dots, Q_m)$ and
a player $i\in Q_j$, we define player $i$'s payoff as 
$$
u_i(S)=\sum_{Q\subseteq Q_j\setminus\{i\}}\frac{|Q|!(|Q_j|-|Q|-1)!}{|Q_j|!}
(v_j(Q\cup \{i\})-v_j(Q)).
$$
Intuitively, the payment to each player can be viewed as his average
payment in the Labor Union model, where the average is taken over
all possible orderings of the players in the party. This immediately
implies $\sum_{i\in Q_j}u_i(S)=v_j(Q_j)$. Thus, Shapley
games share features with both the Fair Value games and the Labor
Union games. Like Fair Value
games, the order in which the players join the party is unimportant. 
Moreover, if all payoff functions are additive, i.e., 
we have $u_i(S\cup\{j\})-u_i(S)=u_i(\{j\})$ for any $i\in N$ and any 
$S\subseteq N\setminus\{i\}$, then the respective Shapley game
coincides with the Fair Value game that corresponds to $(N, v_1, \dots, v_m, M)$. 
On the other hand, similarly to the Labor Union games,
the entire profit of each party is distributed among its members.
We will first show that any Shapley game is an exact potential
game and hence admits a Nash equilibrium in pure strategies
(all proofs in this section are deferred to Appendix~\ref{app:SH}).

\begin{theorem}\label{shapley-potential}
Any Shapley game $\calG=(N, v_1, \dots, v_m, M)$, 
is an exact potential game with the potential function given by
$$
\Phi(S)=\sum_{j\in M}\sum_{Q\subseteq Q_j}
\frac{(|Q|-1)!(|Q_j|-|Q|)!}{|Q_j|!}v_j(Q).
$$
\end{theorem}	
Just like in other profit-sharing games,  
the price of anarchy in Shapley games is bounded by $2$.

\begin{theorem}\label{shapley-poa} 
In any Shapley game $\calG=(N, v_1, \dots, v_m, M)$ with $|N|=n$, we
have $\poa(\calG)\le 2-\frac{1}{n}$.
\end{theorem}
The following claim shows that the bound given in Theorem~\ref{shapley-poa} is almost tight.

\begin{proposition}\label{prop:shapley-poa}
For any $n\ge 3$, 
there exists a Shapley game $\calG=(N, v_1, v_2, M)$ with $|N|=n$ and $|M|=2$
such that $\poa(\calG)=2-\frac{2}{n+1}$ and 
$\pos(\calG)=2-\frac{4}{n+1}$.
\end{proposition}

\section{Cut Games and Profit Sharing Games}\label{sec:cut}
We will now describe a family of succinctly representable
profit-sharing games that can be described in terms of undirected
weighted graphs. It turns out that while two well-studied 
classes of games on such graphs do not induce profit-sharing
games, a ``hybrid'' approach does. We then explain
how to compute players' payoffs in the resulting profit-sharing games.

In the classic {\em cut games}~\cite{Yann91,FabrikantPT04,ChristodoulouMS06}, 
players are the vertices of a weighted graph $G=(N,E)$. The state of the game 
is a partition of players into two parties, and the payoff of each player is the sum of the 
weights of cut edges that are incident on him. 
A cut game naturally corresponds to a coalitional game with the set of players $N$, 
where the value of a coalition $S\subseteq N$ equals to the weight of the cut
induced by $S$ and $N \setminus S$. However, this game is not monotone, 
so it does not induce a profit-sharing game, as defined in 
Section~\ref{sec:main}. 

In {\em induced subgraph games}~\cite{dp94}, the value of a coaliton $S$
equals to the total weight of all edges that have both endpoints in $S$;
while these games are monotone, they are not convex.

Finally, consider a game where the value of a coalition $S\subseteq N$
equals the total weight of all edges incident on vertices in $S$, 
i.e., both internal edges of $S$ (as in induced subgraph games)
and the edges leaving $S$ (as in cut games).
It is not hard to see that this game is both monotone and convex, 
and hence induces a profit-sharing game as described in Section~\ref{sec:main}.
We will now explain how to compute players' 
payoffs in the corresponding Fair Value games, Labor Union games and Shapley games, 
using Figure~\ref{fig:EdgeInterpretation}.
In this figure, we are given a state of the game with two parties $S$
and $N\setminus S$; the players are listed  
from top to bottom in the order in which they (last) entered each party. (The order is 
relevant only in Labor Union games.)  $A$ (resp., $B$) denotes the total weight of edges 
incident on $i$ that connect $i$ to a predecessor (resp., successor) within the party. $C$ is 
the total weight of the cut edges incident on $i$. One can interpret an edge $e=(i,i')$ with 
weight $w(e)$ as a skill or resource of value proportional to $w(e)$ that both $i$ and $i'$ possess. 

\begin{description}
\item[Fair Value Games:] 
The payoff of $i$ (see Figure~\ref{fig:EdgeInterpretation}) is given by $\frac{A+B}{2}+C$. 
Intuitively, the unique skills of a player are weighted more toward his payoff than his shared skills.
\item[Labor Union Games:] 
The payoff of $i$ is given by $B+C$. Intuitively, $i$'s payoff reflects the unique skills that $i$ 
possessed when he joined the party. Players who share skills with $i$, 
but join after $i$, will not get any payoff for those shared skills.
\item[Shapley Games:] One can show that $i$'s payoff is given by $\frac{A+B}{2}+C$, 
just as in Fair Value games.
%
%
\end{description}
One can see that this interpretation easily extends to multiple parties and hyperedges. We also note 
that many of the notions that we have discussed are naturally meaningful in this variant of the cut 
game: for instance, an optimal state for $m=2$ 
is a configuration in which the weighted cut size is maximized. 


\begin{figure}[h]
	\centering
		\includegraphics[height=1.5in]{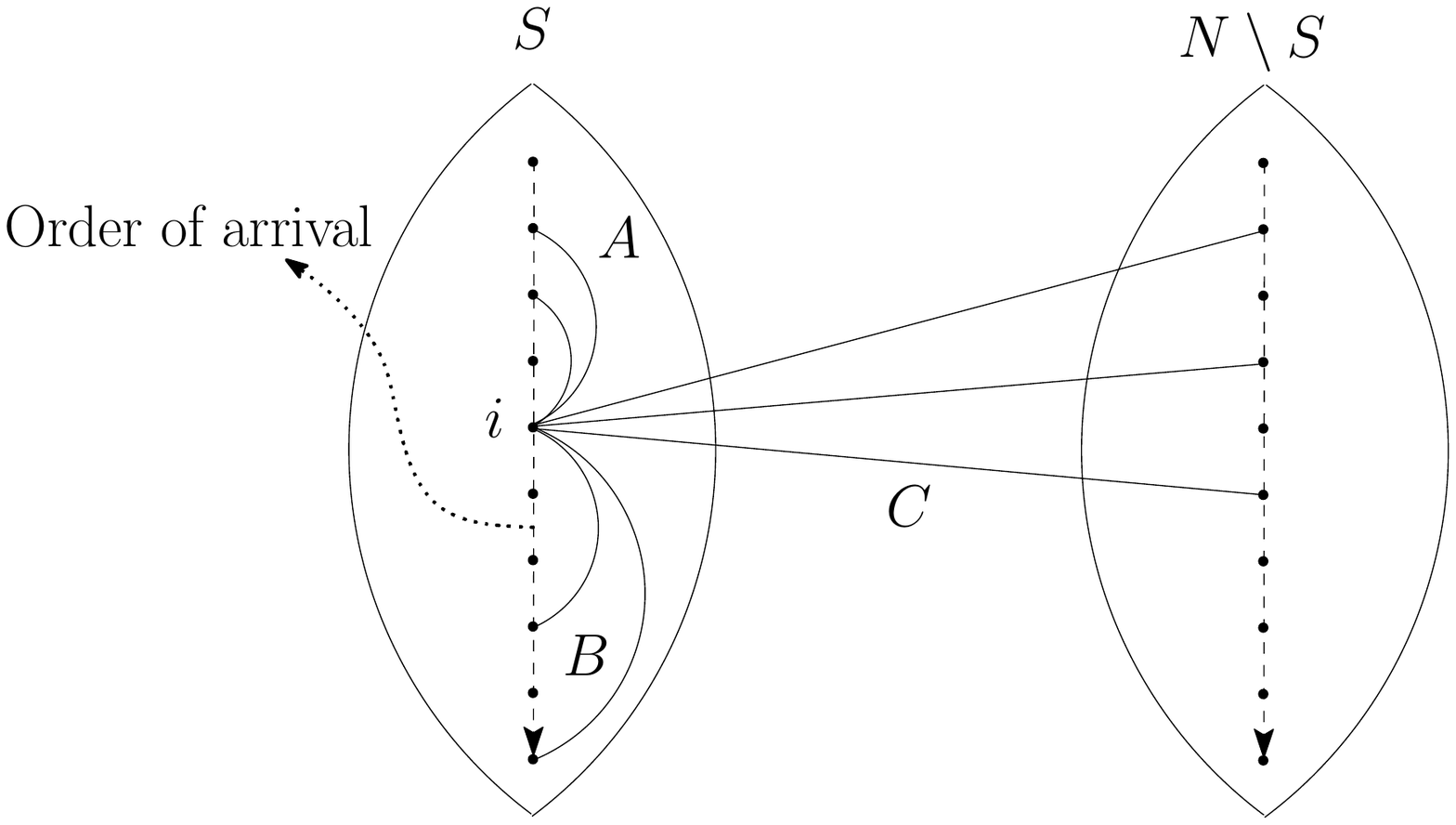}
	\caption{
The set $N$ of players is partitioned into parties $S$ and $N \setminus S$. Consider a player $i$. $A$ (resp., $B$) denotes the total weight of edges incident on $i$ and connecting $i$ to a predecessor (resp., successor) within the party. $C$ is the total weight of the cut edges incident on $i$.}
	\label{fig:EdgeInterpretation}
\end{figure}

\section{Conclusions and Future Work}\label{sec:concl}
In this paper, we studied the dynamics of coalition formation under marginal
contribution-based profit division schemes. 
We have introduced three classes of non-cooperative games that can be constructed
from any convex cooperative game. We have shown that all three profit distribution
schemes considered in this paper have desirable properties: all three games admit 
a Nash equilibrium, and even the worst Nash equilibrium is within a factor of $2$
from the optimal configuration. In addition, for Fair Value games and Labor Union games
a natural dynamic process quickly converges to a state with a fairly high total 
profit. Thus, when rules for sharing the payoff are fixed in advance, 
we can expect a system composed of bounded-rational selfish players to quickly converge to
an acceptable set of teams.

Of course, the picture given by our results is far from complete; 
rather, our work should be seen as
a first step towards understanding the behavior of myopic selfish agents in coaliton
formation settings. In particular, our results seem to suggest that keeping track
of the history of the game and distributing payoffs in a way that respects
players ``seniority'' leads to better stability properties; it would be interesting
to see if this observation is true in practice, and whether it generalizes
to other settings, such as congestion games.

In contrast to the previous work on 
cost-sharing and profit-sharing games, 
our work does not assume that the game's payoffs are given by an
underlying combinatorial structure. Rather, 
our results hold for any convex cooperative game, 
and, in particular, do not depend on whether it is compactly representable. Further, 
all of our results are non-computational in nature. Indeed,  
since the standard representation of cooperative games is exponential in the number
of players, one can only hope to obtain meaningful complexity results 
for subclasses of cooperative games that possess a succinct
representation; identifying such classes and proving complexity results
for them is a promising research direction. 

In our study of Labor Union games, we took a somewhat unusual modeling 
approach: we considered a system described by a sequence of states, each of which induces a 
non-cooperative game, and proved convergence results about the dynamics of such systems.
This approach can be extended to other classes of games such as, e.g., congestion
games; indeed, there are real-life systems where a player's payoff depends
on who selected a certain resource before him. It would be interesting to see
if the known results for congestion games extend to this setting.
  
\newcommand{\etalchar}[1]{$^{#1}$}

\newpage

\appendix
\section{Proofs for Section~\ref{sec:beta}}\label{app:beta}
\subsection{Proof of Lemma~\ref{LM_1}}
From the hypothesis we have
$f(\bar{S}) - f(S) \geq b - \frac{1}{a}f(S)$.
Let $h(S) =  b - \frac{1}{a}f(S)$. Then
$$
h(S) - h(\bar{S}) = \frac{1}{a}( f(\bar{S}) - f(S))\ge \frac{1}{a}h(S).
$$
Hence,
\begin{eqnarray}
h(\bar{S}) \leq \left(1 - \frac{1}{a}\right) h(S). \label{T3}
\end{eqnarray}
Consider a state $S^F$ that is reached by the dynamic starting from
a state $S^I$ in $t$ steps.
By recursively applying~\eqref{T3}, we get
$$
h(S^F) \leq \left(1 - \frac{1}{a}\right)^t  h(S^I).
$$
By setting $t=\lceil a \ln  \frac{h(S^I)}{\epsilon b}\rceil  \leq  \lceil a
\ln\frac{1}{\epsilon}\rceil$ in the previous inequality, we derive that 
$h(S^F) \leq \epsilon b$. 
Thus we obtain 
$f(S^F) = ab\left( 1 - \frac{h(S^F)}{b} \right) 
\geq a b (1 - \epsilon)$.
\qed

\subsection{Proof of Theorem~\ref{TH_2}}
Let us consider a generic state $S=(s_1, \dots, s_n)$ 
of the dynamic. Let $U \subseteq N$ be 
the subset of players that can perform an $\alpha$-best-response 
move, and let  $E = N \setminus U$. Note that no
player $i\in E$ can improve his payoff by more than a factor of $1+\alpha$
by deviating from his current strategy, 
i.e., $\Delta_E(S)\le\alpha\sum_{i\in E}u_i(S)\le \alpha f(S)$.
By definition of a perfect $\beta$-nice game,
we have
$$
 \Delta_E(S) + \Delta_U(S) = \Delta(S)  \geq  \opt_f(\calG) - \beta \cdot f(S).
$$
Let $i$ be the player moving in state $S$, and let $\bar{S}$ be the state resulting
from the move of player $i \in U$. Since $i$ is the player with the maximum absolute
improvement among the players in $U$, we get
\begin{eqnarray*}
 f(\bar{S}) - f(S) & \geq &  \Phi(\bar{S}) - \Phi(S) \\
& \geq & \Delta_i(S) \\
 & \geq & \frac{\Delta_U(S)}{|U|}  \\
 & \geq &  \frac{\opt_f(\calG) - \beta\cdot f(S) - \Delta_E(S)}{n}\\
 & \geq & \frac{\opt_f(\calG) - \beta\cdot f(S) - \alpha \cdot f(S)}{n}\\
& = & \frac{\opt_f(\calG)}{n} - \frac{\beta + \alpha}{n}f(S).
\end{eqnarray*}
The theorem now follows by applying Lemma~\ref{LM_1} with 
$b=\frac{\opt_f(\calG)}{n}$ and $a=\frac{n}{\beta + \alpha}$.
\qed

\section{Proof of Theorem~\ref{th_2}}\label{app:FM}
It is easy to see that $\calG$ is an exact potential game, where 
the potential function is given by the total profit.  In order to prove the theorem, we need 
to show that for each state $S$ we have $2\cdot \tp(S)  +  \Delta(S)  \geq  \opt(\calG)$. 
Consider any state  $S=(s_1, s_2, \ldots, s_n)$, 
and let $S'=(s'_1, s'_2, \ldots, s'_n)$ be the 
state of best responses to $S$, that is, let $s'_i$ be the best 
response of player $i$ in state $S$. Moreover, 
let $S^*=(s^*_1, s^*_2, \ldots, s^*_n)$ be a state that maximizes the total profit. 
Consider a party $k \in M$, and let 
$Q_k=\{i\in N\mid s_i=k\}$, $Q^*_k=\{i\in N\mid s^*_i=k\}$.
We obtain
\begin{eqnarray}
\nonumber 
\Delta_{Q^*_k} (S)& = & 
\sum_{j\in Q^*_k}  \left(\pay_j(S_{-j},s'_j)-\pay_j (S)\right)  \\
& \geq & 
\sum_{j\in Q^*_k}  \left(\pay_j(S_{-j}, k)-\pay_j (S)\right) \label{th_2_1} \\
\nonumber & = & 
\sum_{j\in Q^*_k} \pay_j(S_{-j}, k)-\sum_{j \in Q^*_k} \pay_j (S)  \\
\nonumber & = & 
\sum_{j\in Q^*_k\setminus Q_k}  \left(v_k(Q_k \cup \{j\})-v_k(Q_k)\right) +     
\sum_{j\in Q^*_k \cap Q_k} \pay_j (S)-\sum_{j \in Q^*_k} \pay_j(S)  \\
& \geq &  
v_k(Q_k \cup (Q^*_k\setminus Q_k)) - v_k(Q_k) - \sum_{j\in Q^*_k}\pay_j (S) \label{th_2_2} \\
& \geq &  
v_k(Q^*_k)-v_k(Q_k)-\sum_{j \in Q^*_k} \pay_j(S),\label{th_2_3}
\end{eqnarray}
where~\eqref{th_2_1} holds because for each player $j$ the improvement 
from selecting the best 
response $s'_j$ is at least the improvement achieved by choosing the optimal strategy 
$s^*_j = k$, \eqref{th_2_2} follows from Lemma~\ref{lem_sub}, whereas~\eqref{th_2_3} holds 
because $v_k$ is non-decreasing.

By summing these inequalities over all parties $k$, we obtain
\begin{eqnarray}
\nonumber \Delta(S) = \sum_{k\in M}\Delta_{Q^*_k}(S)
& \geq &   \sum_{k \in M}  v_k(Q^*_k) - \sum_{k \in M} v_k(Q_k)  -    
\sum_{k\in M}  \sum_{j \in  Q^*_k} \pay_j (S)\\
\nonumber &  =  &   \tp(S^*) -  \tp(S) - \sum_{j \in N} \pay_j (S)\\
& \geq & \opt(\calG) - 2 \tp(S). \label{th_2_sol}
\end{eqnarray}
where \eqref{th_2_sol} follows from the fact that for every state $S$ we have 
$ \sum_{j \in N} \pay_j (S) \leq \tp(S)$.
\qed

\section{Proofs for Section~\ref{sec:SH}}\label{app:SH}

\subsection{Proof of Theorem~\ref{shapley-potential}}
Suppose that in some state $S=(Q_1, \dots, Q_m)$ of the game a player
$i$ that belongs to party $1$ wants to switch to party $2$.   
Let $S'$ be the state after player $i$ switches.
Our goal is to show that $u_i(S')-u_i(S)=\Phi(S')-\Phi(S)$,
so $\Phi$ is indeed a potential  function of the game.

We can compute the utility of player $i$ in both states,
taking into account that in state $S$ player $i$ belongs to party~$1$ with $|Q_1|$ members,
but in state $S'$ she belongs to party~$2$ with $|Q_2+1|$ members:
\begin{eqnarray*}
u_i(S')    &=& \sum_{Q\subseteq Q_2}
              \frac{|Q|!(|Q_2|-|Q|)!}{(|Q_2|+1)!}(v_2(Q\cup \{i\})-v_2(Q)), \\
u_i(S) &=& \sum_{Q\subseteq Q_1\setminus\{i\}}
              \frac{|Q|!(|Q_1|-|Q|-1)!}{|Q_1|!}(v_1(Q\cup \{i\})-v_1(Q)).
\end{eqnarray*}
The only parties whose composition changes as we move from state $S$
to state $S'$ are party $1$ and party $2$. Therefore, when 
computing the difference between $\Phi(S')$ and $\Phi(S)$,
we can ignore all other parties:
\begin{eqnarray*}
&&\Phi(S')-\Phi(S)  =
\sum_{Q\subseteq Q_1\setminus \{i\}}
      \frac{(|Q|-1)!(|Q_1|-1-|Q|)!}{(|Q_1|-1)!}v_1(Q) 
      \\
      && \quad +
\sum_{Q\subseteq Q_2\cup \{i\}}
      \frac{(|Q|-1)!(|Q_2|+1-|Q|)!}{(|Q_2|+1)!}v_2(Q)\\
&& \quad - \sum_{Q\subseteq Q_1}
      \frac{(|Q|-1)!(|Q_1|-|Q|)!}{|Q_1|!}v_1(Q)
      \\
      && \quad -
     \sum_{Q\subseteq Q_2}
      \frac{(|Q|-1)!(|Q_2|-|Q|)!}{|Q_2|!}v_2(Q)\\
& = &\sum_{Q\subseteq Q_2}
      \left(\left(\frac{(|Q|-1)!(|Q_2|+1-|Q|)!}{(|Q_2|+1)!}-
      \frac{(|Q|-1)!(|Q_2|-|Q|)!}{|Q_2|!}\right)v_2(Q) \right. \\ 
      && \quad \quad \quad \quad + \left.
      \frac{(|Q|)!(|Q_2|-|Q|)!}{(|Q_2|+1)!}v_2(Q\cup \{i\})\right)\\
&& \quad + \sum_{Q\subseteq Q_1\setminus\{i\}}
      \left(\left(\frac{(|Q|-1)!(|Q_1|-1-|Q|)!}{(|Q_1|-1)!} - \frac{(|Q|-1)!(|Q_1|-|Q|)!}{|Q_1|!}\right)v_1(Q) \right. \\
      &&  \quad \quad \quad \quad - \left.
       \frac{(|Q|)!(|Q_1|-|Q|-1)!}{|Q_1|!}v_1(Q\cup \{i\})\right)\\
& = &\sum_{Q\subseteq Q_2}
          \frac{|Q|!(|Q_2|-|Q|)!}{(|Q_2|+1)!}(v_2(Q\cup \{i\})-v_2(Q))\\
          && \quad -
     \sum_{Q\subseteq Q_1\setminus\{i\}}
          \frac{|Q|!(|Q_1|-|Q|-1)!}{|Q_1|!}(v_1(Q\cup\{i\})-v_1(Q))\\
& = &u_i(S')-u_i(S).
\end{eqnarray*}
\qed

\subsection{Proof of Theorem~\ref{shapley-poa}}
Let $S=(Q_1, \dots, Q_m)$ be a Nash equilibrium state, and let $S^*=(Q_1^*, \dots, Q^*_m)$ 
be a state where the maximum total profit is achieved. 
It suffices to show that $(1-\frac{1}{n})\tp(S) + \tp(S)\geq \tp(S^*)$.

Observe first that if $|Q_j|=n$ for some $j\in M$, then $S$ 
is an optimal state. Indeed, if $S$ is not optimal, 
by the total payoff distribution 
property there exists a party $k\in M$ 
and a player $i\in Q^*_k$ such that $u_i(S^*)>u_i(S)$.
If player $i$ switches to party $k$, 
which currently has no members, by submodularity property his
payoff will be at least $u_i(S^*)$, a contradiction with $S$
being a Nash equilibrium state. Therefore, from now on, 
we assume that $|Q_j|<n$ for all $j\in M$.

Now, we have
$$
\tp(S)=\sum_{i\in N} u_i(S) =\sum_{j\in M}\sum_{i\in Q^*_j} u_i(S).
$$ 
For any $j\in M$ and all $i\in Q^*_j$, we can derive a lower bound
on $u_i(S)$. There are two cases to be considered.

\begin{itemize}
\item[(1)] 
If $i\in Q_j$, we have
\begin{eqnarray*}
u_i(S) 
&=& \sum_{Q \subseteq Q_j \setminus \{i\}}
           \frac{|Q|!(|Q_j|-|Q|-1)!}{|Q_j|!}(v_j(Q\cup \{i\}) - v_j(Q))\\  
&>& \sum_{Q \subset Q_j \setminus \{i\}}
           \frac{|Q|!(|Q_j|-|Q|)!}{(|Q_j|+1)!}(v_j(Q\cup \{i\}) - v_j(Q)).
\end{eqnarray*}

\item[(2)]
If $i\notin Q_j$, we have
$$
u_i(S) 
\geq  \sum_{Q \subseteq Q_j }\frac{|Q|!(|Q_j|-|Q|)!}{(|Q_j|+1)!}(v_j(Q\cup \{i\}) - v_j(Q)), 
$$
since $S$ is a Nash equilibrium, and hence player $i$ cannot increase his utility
by switching to party $j$. 
\end{itemize}
Changing the order of summation, by Lemma~\ref{lem_sub}, we have
$$
\tp(S) \geq \sum_{j\in M}\sum_{Q\subseteq Q_j}
            \frac{|Q|!(|Q_j|-|Q|)!}{(|Q_j|+1)!}(v_j(Q\cup Q^*_j) - v_j(Q)).
$$
Set $q=|Q_j|$. We have
$$
\sum_{Q\subseteq Q_j}\frac{|Q|!(|Q_j|-|Q|)!}{(|Q_j|+1)!}=
\sum_{i=0}^q\sum_{Q\subseteq Q_j, |Q|=i}\frac{|Q|!(|Q_j|-|Q|)!}{(|Q_j|+1)!}=
\sum_{i=0}^q{q\choose i}\frac{i!(q-i)!}{(q+1)!}=\sum_{i=0}^q\frac{1}{q+1}=1;
$$ 
this identity can also be derived by considering Shapley values in an additive game
with $|Q_j|+1$ players. Further, we have $v_j(Q\cup Q^*_j) \geq v_j(Q^*_j)$.
Thus,  
\begin{eqnarray}\label{eq:shapley1}
\tp(S)\geq \sum_{j\in M}  v_j(Q^*_j) - 
\sum_{j\in M}\sum_{Q\subseteq Q_j}\frac{|Q|!(|Q_j|-|Q|)!}{(|Q_j|+1)!}v_j(Q).
\end{eqnarray}
For any $Q\subseteq Q_j$, we have $v_j(Q)\leq v_j(Q_j)$, and, moreover, $v_j(\emptyset) = 0$. 
Recall also that we assume that $|Q_j|<n$ for all $j\in M$.
Thus we can bound the negative term in the right-hand side of~\eqref{eq:shapley1} as
\begin{eqnarray}\label{eq:shapley2}
\sum_{j\in M}
\sum_{Q\subseteq Q_j, Q\neq\emptyset}\frac{|Q|!(|Q_j|-|Q|)!}{(|Q_j|+1)!}v_j(Q_j) 
= 
\sum_{j\in M}  (1-\frac{0!(|Q_j|-0)!}{(|Q_j|+1)!})v_j(Q_j)
\leq (1-\frac{1}{n})\tp(S).  
\end{eqnarray}
Combining~\eqref{eq:shapley1} and~\eqref{eq:shapley2}, 
we obtain $(2-1/n) \tp(S) \geq \tp(S^*)$.
\qed

\subsection{Proof of Proposition~\ref{prop:shapley-poa}}
\begin{proof}
Let $v_1$ be an additive function given by
$v_1(\{1\})=\frac{1}{n}$, 
$v_1(\{i\})=\frac{1}{n-1}$ for $i\geq 2$, 
and let $v_2(Q)=1$ for any $Q\neq\emptyset$. 

The state $S^*=(Q^*_1, Q^*_2)$
with $Q^*_1=\{2,\ldots,n\}$, $Q^*_2=\{1\}$
has  total profit
$(n-1)\frac{1}{n-1}+1=2$, which is the optimum in this game.

On the other hand, a state $S=(Q_1, Q_2)$ 
with $Q_1=\{1\}$, $Q_2=\{2,\ldots,n\}$ is a Nash equilibrium. 
Indeed, player $1$ is paid $1/n$ and will be paid the same amount
if he switches parties, so he has no 
incentive to switch. All other players are paid $\frac{1}{n-1}$, 
and any of them will be paid 
the same amount if he switches to the first party. Therefore none of them 
has an incentive to switch either. 
The total profit in state $S$ is $1+\frac{1}{n}$.
There is no Nash equilibrium with a smaller total profit, 
because in any Nash equilibrium state there are 
players in both parties, and hence the total profit is at least $\frac{1}{n}+1$.
Thus, $\poa(\calG)=\frac{2}{1+1/n}=2-\frac{2}{n+1}$.

In any Nash equilibrium, party $2$ contains at least $n-2$ players.
Hence, the total profit in any Nash equilibrium is at most $\frac{2}{n-1}+1$.
This profit is achieved in, e.g., state $S'=(Q_1', Q_2')$ with 
$Q'_1=\{n-1,n\}$, $Q'_2=\{1,\ldots,n-2\}$. 
Therefore, $\pos(\calG)=\frac{2}{1+2/(n-1)}=2-\frac{4}{n+1}$.
\qed
\end{proof}

\end{document}